\documentclass[11pt]{article}

\usepackage{amsfonts,mathrsfs,commath}
\usepackage{dsfont,color,tikz,tkz-graph,enumerate,placeins,multicol}
\usetikzlibrary{arrows.meta}

\usepackage{graphicx}
\usepackage{mathtools}
\usepackage{amsmath}
\RequirePackage{amssymb}
\RequirePackage{amsthm}
\usepackage{amsmath,scalefnt}
\usepackage{footnote}
\usepackage{fmtcount}
\usepackage{threeparttable}
\usepackage{caption}
\usepackage{subcaption}
\usepackage{color}
\usepackage{tikz}
\usetikzlibrary{positioning,chains,fit,shapes,calc}
\usetikzlibrary{arrows}
\tikzstyle{vertex}=[circle, draw, inner sep=0pt, minimum size=6pt]
 \usepackage{enumitem}
 \usepackage{multicol}
 \usepackage{wrapfig}
\usepackage{enumitem}

\newcommand{\notforthesis}[1]{}

	\newtheorem{definition}{Definition}
	\newtheorem{claim}{Claim}
	\newtheorem{lemma}{Lemma}
         
         \newtheorem{example}{Example}

	 \newtheorem{assumption}{Assumption}
	 \newtheorem{proof of proposition}{Proof of Proposition}
	\newtheorem{proof of claim}{Proof of Claim}
	\newtheorem{proposition}{Proposition} 
	\newtheorem{theorem}{Theorem}
	\newtheorem{notation}{Notation}
        \newtheorem{corollary}{Corollary}
        \newtheorem{remark}{Remark}

\newcommand{\diag}{\mbox{diag}\,}

\renewcommand{\r}{{\mathbb R}}
\renewcommand{\a}{{\alpha}}

\newcommand{\e}{{\epsilon}}

\newcommand{\D}{{\color{black}D}}
\newcommand{\M}{\mu}

\newcommand{\ee}{\end{equation}}
\newcommand{\bal}{\begin{aligned}}
\newcommand{\eal}{\end{aligned}}
\newcommand{\bi}{\begin{itemize}}
\newcommand{\ei}{\end{itemize}}
\newcommand{\ben}{\begin{enumerate}}
\newcommand{\een}{\end{enumerate}}
\newcommand{\beqn}{\begin{eqnarray*}}
\newcommand{\eeqn}{\end{eqnarray*}}
\newcommand{\pic}[2]{\includegraphics[scale=#1]{#2}}

\newcommand{\be}[1]{\begin{equation}\label{#1}}
\newcommand{\bp}{\begin{proof}}
\newcommand{\ep}{\end{proof}}
\newcommand{\bremark}{\begin{remark}\rm } 
\newcommand{\eremark}{\end{remark}}
\newcommand{\blem}{\begin{lemma}}
\newcommand{\elem}{\end{lemma}}
\newcommand{\bclaim}{\begin{claim}}
\newcommand{\eclaim}{\end{claim}}
\newcommand{\bass}{\begin{assumption}}
\newcommand{\eass}{\end{assumption}}
\newcommand{\bnote}{\begin{notation}}
\newcommand{\enote}{\end{notation}}

\newcommand{\bthm}{\begin{theorem}}
\newcommand{\ethm}{\end{theorem}}
\newcommand{\bprop}{\begin{proposition}}
\newcommand{\eprop}{\end{proposition}}
\newcommand{\bcor}{\begin{corollary}}
\newcommand{\ecor}{\end{corollary}}
\newcommand{\dis}{\displaystyle}
\newcommand{\lt}{\left}
\newcommand{\rt}{\right}

\title{Cluster synchronization of diffusively-coupled \\nonlinear systems: A contraction based approach}
\author{
Zahra Aminzare%
\thanks{The Program in Applied and Computational Mathematics, Princeton University, NJ 08544, USA. Email: aminzare@math.princeton.edu (Z. Aminzare).}
\and Biswadip Dey$^{\dag}$%
\and Elizabeth N. Davison$^{\dag}$%
\and Naomi Ehrich Leonard%
\thanks{Department of Mechanical and Aerospace Engineering, Princeton University, NJ 08540, USA. Emails: biswadip@princeton.edu (B. Dey), end@princeton.edu (E. N. Davison), naomi@princeton.edu (N. Ehrich Leonard).}
}
\date{}                                           

\begin{document}
\maketitle

\begin{abstract}
Finding the conditions that foster synchronization in networked oscillatory systems is
critical to understanding a wide range of biological and mechanical systems. However, the conditions proved in the literature for synchronization in nonlinear systems with linear coupling, such as has been used to model neuronal networks, are in general not strict enough to accurately determine the system behavior. We leverage contraction theory to derive new sufficient conditions for cluster synchronization in terms of the network structure, for a network where the intrinsic nonlinear dynamics of each node may differ. Our result requires that network connections satisfy a cluster-input-equivalence condition, and we explore the influence of this requirement on network dynamics. For application to networks of nodes with neuronal spiking dynamics, we show that our new sufficient condition is tighter than those found in previous analyses which used nonsmooth Lyapunov functions. Improving the analytical conditions for when cluster synchronization will occur based on network configuration is a significant step toward facilitating understanding and control of complex oscillatory systems.
\end{abstract}

\section{Introduction}

Synchronization has been observed and studied in diverse fields. Its presence has been characterized in symmetric networks of identical mechanical systems or identical biological systems, as well as those with differing types of individual components and nonuniform coupling \cite{Sync_Book_NL, strogatz2003sync}. The role of synchronization has been studied in a multitude of both natural and engineered settings including collective motion \cite{Synchro_RS_DP_NEL}, power-grid networks \cite{Sync_Power_Grid_Nat_Phys}, robotics \cite{Sync_Mesh_System_Networks}, sensor networks \cite{Time_Sync_Sensor_Net_Survey}, circadian rhythms \cite{Sync_Circadian}, bioluminescence in fireflies \cite{Sync_Firefly_Firing_Smith1935}, pacemaker cells in the heart \cite{Sync_Pacemaker_Cells_Strogatz}, neuronal ensembles \cite{Sync_Neuronal_Networks}, and numerous others. In the human brain, synchronization at the neuronal or regional level can be beneficial, allowing for production of a vast range of behaviors \cite{Sync_Brain_Social_Interact, Sync_OdorEncoding_Neural_System}, or detrimental, causing disorders such as Parkinson's disease \cite{lehnertz2009synchronization} and epilepsy \cite{chen2007excessive}. Applications for control of neural dynamics may involve regulating patterns of synchronized phenomena among nodes or subsystems that have different intrinsic dynamics and are connected in an arbitrary network \cite{abrams2016introduction, wilson2015clustered}. Most generally, nodes can be agents in a multi-agent system, compartments in a compartmental system, or other units that interact with one another in a pairwise framework. Characterizing the emergence and persistence of synchronization in a system with multiple heterogeneous nodes is the first step towards effective control of desired behavior.

In realistic networks that feature heterogeneous nodes and nonuniform coupling structure, complex patterns of synchronization emerge. Under certain conditions, it is possible to partition the network into clusters of nodes that are synchronized within clusters but not across clusters. This is called cluster synchronization \cite{Cluster_synchronization_Belykh_2008, sorrentino_network_2007}. In the field of pattern formation, the formation of clusters has been investigated extensively \cite{Cluster_Sync_Patterning_4, Cluster_Sync_Patterning_3}. The conditions for synchronized clusters can be approached analytically by generalizing approaches from the literature on synchronization \cite{Cluster_Sync_NonIdentical_DS, Cluster_Sync_Sorrentino, Cluster_Sync_Community, Cluster_Sync_MingCao}. Here, we leverage contraction theory to provide new sufficient conditions for synchronization of clusters in a network with heterogeneous oscillators.

Contraction theory is a powerful tool for understanding synchronization phenomena in networked systems. The proper tool for characterizing contractivity for nonlinear systems is provided by the logarithmic norms, or matrix measures \cite{michelbook,%
Desoer}, of the Jacobian of the vector field, evaluated at all possible states. This idea is a classical one, and can be traced back at least to work of D.C.~Lewis in the 1940s, \cite{Lewis1949,%
Hartman1961}. Dahlquist's 1958 thesis under H\"ormander used matrix measures to show contractivity of differential equations, and more generally of differential inequalities, the latter applied to the analysis of convergence of numerical schemes for solving differential equations \cite{dahlquist}. Several authors have independently rediscovered the basic ideas. For example, in the 1960s, Demidovi{\v{c}}~\cite{Demidovich1961,%
Demidovich1967} established basic convergence results with respect to Euclidean norms, as did Yoshizawa~\cite{Yoshizawa1966,%
Yoshizawa1975}. In control theory, the field attracted much attention after the work of Lohmiller and Slotine~\cite{lohmiller1998contraction}.
We refer the reader especially to the careful historical analysis given in \cite{Jouffroy}. Other useful historical references are~\cite{Pav_Pog_Wou_Nij} and the survey~\cite{Soderlind}.
{%
An introductory tutorial to basic results in contraction theory for nonlinear control systems is given in~\cite{contraction_survey_CDC14}.
Results on synchronization using contraction-based techniques are described, for example, in \cite{Slotine_cluster_synch, 
arcak2011,%
slotine,%
Russo_Bernardo,%
slotine_wang,%
aminzare_shafi_arcak_sontag_bookchapter2013}.

The main contributions of the present paper are as follows. We extend contraction theory to a setting where oscillators may have heterogeneous intrinsic dynamics and the network satisfies the cluster-input-equivalence condition. Using this extension of contraction theory, we prove new sufficient conditions for cluster synchronization in a network of heterogeneous oscillators. Furthermore, we improve upon our earlier analysis of synchronization in networks of homogeneous Fitzhugh-Nagumo (FN) oscillators \cite{davison_sync_2016}, and show that the proposed result yields a tighter bound on the algebraic connectivity of the associated undirected graph.

The paper proceeds as follows. In Section~\ref{contraction_review}, we review relevant concepts and results from the contraction theory literature. We present our main result, an extension of the existing theory to a cluster synchronized setting, in Section~\ref{cluster_synchronization}. In Section~\ref{Network_Reduction}, we demonstrate how we can use cluster synchronization to reduce a large network of nodes with heterogeneous intrinsic dynamics into a smaller network of their synchronized states. Finally, in Section~\ref{Application}, we consider a network of FN oscillators and demonstrate how the contraction based approach provides improvement over existing results on sufficient conditions for synchronization and cluster synchronization.

\section{Contraction Theory: Review}
\label{contraction_review}

In what follows, we review notations, definitions, and main results in contraction theory that will be applied in later sections.

\begin{definition}[Logarithmic norm \cite{Soderlind}]
For any matrix $A\in\r^{n\times n}$ and any given norm $\|\cdot\|$ on $\r^n$, the logarithmic norm (also called the matrix measure) of $A$ induced by the norm $\|\cdot\|$ is defined by
\be{logarithmic_norm}
\M [A] = \displaystyle\lim_{h\to0^+}\sup_{x\neq 0 \in \r^n}\frac{1}{h}\left(\frac{\|(I+hA)x\|}{\|x\|}-1\right), 
\end{equation}
where $I$ is the identity matrix of size $n$. 
\end{definition}

\bnote
For any $1\leq p\leq\infty$ and any $n\times n$ positive definite matrix $Q$, let $\|\cdot\|_p$ denote the $L^p$ norm on $\r^n$, and 
$\|\cdot\|_{p,Q}$ denote the $Q-$weighted $L^p$ norm on $\r^n$ defined by $\|x\|_{p,Q}:=\|Qx\|_p$. 
By $\M_{p}[A]$, we mean the logarithmic norm of $A$ induced by $\|\cdot\|_p$ and 
by $\M_{p,Q}[A]$, we mean the logarithmic norm of $A$ induced by $\|\cdot\|_{p,Q}$. Note that $\M_{p,Q}[A] = \M_p[QAQ^{-1}].$
\enote

\bremark\label{explicit_value_LN}
 In Table \ref{tab-mu}, the algebraic expression of logarithmic norms induced by the $L^p$ norm for $p=1,2,$ and $\infty$ are shown. For proofs, see for instance \cite{Desoer}.
\newcommand{\spec}{\mbox{spec}}
\begin{table}[ht]
{\scriptsize{
\caption{\scshape Standard matrix measures for a real $n\times n$ matrix, $A=[a_{ij}]$.}
\begin{center}
\begin{tabular}{|c|c|}
\hline
vector norm, $\|\cdot\|$ & induced matrix measure, $M[A]$\\
\hline
$\|x\|_1=\displaystyle\sum_{i=1}^{n}\abs{x_i}$ & $\M_1[A]=\displaystyle\max_{j}\left(a_{jj}+\displaystyle\sum_{i\neq j}\abs{a_{ij}}\right)$\\
\hline
$\|x\|_2=\left(\displaystyle\sum_{i=1}^{n}\abs{x_i}^2\right)^{\frac{1}{2}}$ & $\M_2[A]=\displaystyle\max_{\lambda\in\spec{\frac{1}{2}(A+A^T)}}
\lambda$\\
\hline
$\|x\|_{\infty}=\displaystyle\max_{1\leq i\leq n} \abs{x_i}$ & $\M_{\infty}[A]=\displaystyle\max_{i}\left(a_{ii}+\displaystyle\sum_{i\neq j}\abs{a_{ij}}\right)$\\
\hline
\end{tabular}
\end{center}
\label{tab-mu}
}}
\end{table}%
\eremark

\begin{definition}[Contraction]
Consider the following nonlinear dynamical system on $V\times[0,\infty]$, where $V$ is a convex subset of $\r^n$. Consider appropriate conditions on vector field $G$ (e.g. $G(x,t)$ Lipschitz on $x$ and continuous on $(x,t)$) which guarantee existence and uniqueness of solutions of
\be{isolated_ODE}
\dot x(t) = G(x(t),t). 
\ee
Equation~(\ref{isolated_ODE}) is {\em contractive} if there exist $c<0$ and a norm $\|\cdot\|$ on $\r^n$ such that, for any two solutions $x$ and $y$ of Equation~(\ref{isolated_ODE}), the following inequality holds for any $t\geq0$:
\be{contractivity_inequality}
\|x(t)-y(t)\| \leq e^{ct} \|x(0) - y(0)\|. 
\ee
\end{definition}

\bprop[Theorem 1, \cite{contraction_survey_CDC14}]
\label{general_contraction}
Consider Equation~(\ref{isolated_ODE}) and assume that $G$ is a continuously differentiable function on its first variable. Let $c:= \sup_{(x,t)}\mu [J_G(x,t)]$, where $\mu$ is the logarithmic norm induced by an arbitrary norm on $\r^n$, and $J_G$ is the Jacobian of $G$. Then for any two solutions $x$ and $y$ of Equation~(\ref{isolated_ODE}),
\[
\|{x(t) - y(t)}\| \leq e^{ct} \|{x(0) - y(0)}\|.
\]
In particular, when $c<0$, Equation~(\ref{isolated_ODE}) satisfies Equation~(\ref{contractivity_inequality}) and is contractive. 
\eprop

Throughout the paper, we denote the Jacobian of the vector field $f(x,t)$ evaluated at $(x,t)$ as $J_{f}(x,t)$, i.e., $J_{f}(x,t)=\frac{\partial f}{\partial x}(x,t)$.

We consider a network of $N$ nodes, with states $\{X^1, \ldots, X^N\}$ and intrinsic dynamics $F^i$:
\beqn
\dot X^i(t) = F^i\lt(X^i(t),t\rt)\;.
\eeqn 
Here, $X^i$ and $F^i$ have dimension $n\geq1$. For a fixed convex subset $V\subset\r^n$, $F^i \colon V \times [0,\infty) \to\r^{n}$, defined by $F^i=F^i(z,t)$, is Lipschitz on $z$ and continuous on $(z,t)$. We also assume that the nodes are diffusively connected through an undirected weighted graph $\mathcal{G} = (\mathcal{V}, \mathcal{E})$ and describe the dynamics of the network as follows:
\begin{align}\label{general_network}
\dot{X}^i(t) &= F^i\lt(X^i(t),t\rt)+\sum_{j\in \mathcal{N}^i} \gamma^{ij} \D \lt(X^j(t)-X^i(t)\rt)\qquad i=1,\ldots,N\;.
\end{align}
The indices in $\mathcal{N}^i$ represent the neighbors of node $i$. The \textit{diffusion matrix} $\D$ is a nonzero diagonal matrix of size $n$, $\D=\diag(d_1,\ldots, d_n),$ where $d_i\geq0$. The positive constants $\gamma^{ij}$ represent the edge weights of $\mathcal{G}$. The products of the elements in $\D$ and the edge weights $\gamma^{ij}$ represent the coupling strengths between the nodes.

Let $\mathcal{L} = (\mathcal{L}_{ij})$ be the Laplacian matrix of $\mathcal{G}$:
\be{Laplacian_mathcal_L}
\mathcal{L}_{ij} = \left\{\begin{array}{ccc}
\sum_{k\in \mathcal{N}^i } \gamma^{ik} &  & i=j ,\\
-\gamma^{ij} &  & i\neq j, j\in \mathcal{N}^i,  \\
0 &  & \mbox{otherwise}.
\end{array}\right.
\ee
We denote the eigenvalues of $\mathcal{L}$ as $0=\lambda^{(1)} \leq \lambda^{(2)} \leq \cdots \leq \lambda^{(N)}$. The second smallest eigenvalue, $\lambda^{(2)}$, is called the {\em algebraic connectivity} of the graph. This number helps to quantify ``how connected" the graph is. The number of the zero eigenvalues is equal to the number of connected components of $\mathcal G$.

Using the notation of the Laplacian matrix, Equation~(\ref{general_network}) can be written in the following closed form:
\begin{align}\label{general_network_closed_form1}
\dot{X} (t)&= \mathcal F(X(t), t) - (\mathcal L\otimes \D) X(t), 
\end{align}
where $X=\lt({X^1}^T,\ldots, {X^N}^T\rt)^T$, $\mathcal F=\lt({F^1}^T,\ldots, {F^N}^T\rt)^T$, and $\otimes$ represents the Kronecker product. 

\begin{definition}[Complete synchronization]\label{def_sync_manifold} 
Let 
\[
\mathscr{S}_1:= \left\{X^1=\cdots=X^N, \quad X^i\in\mathbb{R}^n \right\}. 
\]
The dynamics given in Equation~(\ref{general_network}) {\em synchronize completely} if any solution of Equation~(\ref{general_network}) converges to $\mathscr{S}_1$ in an appropriate norm. 
In other words, let $X$ be a solution of Equation~(\ref{general_network}). Then there exists a solution $\bar X \in \mathscr{S}_1$ such that, in an appropriate norm,
\[
X(t)-\bar X(t)\to0  \quad \mbox{as $t\to\infty.$}
\]
$\mathscr{S}_1$ is called the \em{synchronization manifold}.
\end{definition}
We will use synchronization and complete synchronization alternatively. 

\begin{definition}[Cluster synchronization]\label{Cluster_Synchronization}
For any $1\leq K\leq N$ and any $1\leq c_1, \ldots, c_K \leq N$ such that $c_1+\cdots+c_K=N$, let 
\[
\mathscr{S}_K:= \left\{X^1=\cdots=X^{c_1}, \; \ldots,\;    X^{N- c_K+1}=\cdots=X^{N}, \quad X^i\in\mathbb{R}^n\right\}. 
\]
The dynamics given in Equation~(\ref{general_network}) {\em synchronize in clusters} if there exists $1\leq K\leq N$ such that   any solution of Equation~(\ref{general_network}) converges to $\mathscr{S}_K$ in an appropriate norm. 

$\mathscr{S}_K$ is called the \em{$K-$cluster synchronization manifold}.
\end{definition}
Note that, the 1-cluster synchronization manifold is same as the synchronization manifold (Definition~\ref{def_sync_manifold}).

In the following two propositions, we consider Equation~(\ref{general_network}) with homogeneous $F^i=F$, 
and state two sufficient conditions that guarantee that Equation~(\ref{general_network}) synchronizes. 

\bprop[Proposition 1, \cite{synchronization2014_journal}]
\label{cor_condition_contraction}
Consider Equation~(\ref{general_network}) with homogeneous $F^i=F$. Assume that there exists a norm on $\r^n$ such that 
\be{condition_contraction}
 \sup_{(x,t)} \M[J_F(x,t)]<0. 
\ee 
Then Equation~(\ref{general_network}) synchronizes. 
\eprop
In \cite{Slotine_cluster_synch}, Proposition \ref{cor_condition_contraction} has been generalized\footnote{The statement of Theorem 3 in \cite{Slotine_cluster_synch} is correct; however, the proof needs revision to be complete.} to $F^i$ with heterogeneous elements. The work shows that, under some conditions on the weights of the interconnected graph, if each node has contractive dynamics, then Equation~(\ref{general_network}) synchronizes in clusters. In Section \ref{Application}, we provide an example that synchronizes in clusters and supports our theory derived in the next section but does not satisfy the condition provided in \cite{Slotine_cluster_synch}.

Note that the sufficient condition provided in Proposition~\ref{cor_condition_contraction} depends only on the dynamics of each isolated node, namely $J_F$. 
The next proposition provides a sufficient condition for synchronization weaker than Equation~(\ref{condition_contraction}) that depends on $J_F$, the diffusion matrix $D$, and the graph $\mathcal{G}$. 
The following results are based on weighted $L^2$ norms but, for some special graphs, they have been generalized to weighted $L^p$ norms \cite{synchronization2014_journal}.
\begin{proposition}[Theorem 4 (modified), \cite{arcak2011}]
\label{synchronization_Arcak}
Consider Equation~(\ref{general_network}) with homogeneous $F^i=F$. Assume that there exists a positive definite matrix $P$ such that $P^2D+DP^2$ is also positive definite, and let
\[
c:= \sup_{(x,t)\in V\times [0,\infty)} \M_{2,P}\lt[J_{F}(x,t)-\lambda^{(2)} D\rt].
\]
Then for any solution $X$ of Equation~(\ref{general_network}) that remains in $V^N$, there exists a solution $\bar X$ such that
\[
\| X(t) - \bar X(t) \|_{2,P} \leq e^{ct} \| X(0) - \bar X(0) \|_{2,P}. 
\]
Moreover, if $c<0$, then Equation~(\ref{general_network}) synchronizes, i.e., for any pair $i, j \in \{1,\ldots, N\}$, 
\[
X^i(t) - X^j(t) \to 0 \quad \mbox{as $t\to\infty$}. 
\]
\eprop

In the following section, we present the main result of this work -- we generalize Proposition~\ref{synchronization_Arcak} to heterogeneous $F^i$ and provide sufficient conditions for cluster synchronization. 

\section{Main Result: Cluster Synchronization}
\label{cluster_synchronization}

In this section, we provide sufficient conditions on heterogeneous intrinsic dynamics $F^i$, 
the graph $\mathcal{G}$, and the diffusion matrix $\D$, that guarantee cluster synchronization of the network described in Equation~(\ref{general_network}). 

\bass\label{Assumption}
In the network described by Equation~(\ref{general_network}), we assume that
\begin{enumerate}
\item There exist $K\leq N$ and $c_1, \ldots, c_K\geq 2$, such that $c_1+\cdots+c_K=N$, and 
\[F^{i_1}=\cdots=F^{i_{c_1}} =: F_{ \mathscr{C}_1}, \; \ldots,\;    F^{i_{N- c_K+1}}=\cdots=F^{i_N} =: F_{ \mathscr{C}_K}, \]
where $\{i_1,\ldots, i_N\}$ is a permutation of $\{1,\ldots, N\}$. 
Without loss of generality, we can assume: 
\[F^{1}=\cdots=F^{{c_1}} =: F_{ \mathscr{C}_1}, \; \ldots,\;    F^{{N- c_K+1}}=\cdots=F^{N} =: F_{ \mathscr{C}_K}.\]
Let $\mathscr{C}_1, \ldots, \mathscr{C}_K$ denote $K$ clusters of nodes. The nodes in cluster $\mathscr{C}_1$ are defined by  $X^1,\ldots, X^{c_1}$ and they all have dynamics $F_{ \mathscr{C}_1}$, the nodes in cluster $\mathscr{C}_2$ are defined by $X^{c_1+1}, \ldots, X^{c_1+c_2}$ and they all
have dynamics $F_{ \mathscr{C}_2}$, etc. For ease of notation in our calculations, we let 
 \be{}
\bal
X^1_{ \mathscr{C}_1} = X^1,&\ldots, X^{c_1}_{ \mathscr{C}_1}= X^{c_1},\\
X^1_{ \mathscr{C}_2}=X^{c_1+1} ,&\ldots, X^{c_2}_{ \mathscr{C}_2}= X^{c_2},\\
&\;\; \vdots\\ 
X^1_{ \mathscr{C}_K} =X^{N-c_K+1},&\ldots, X^{c_K}_{ \mathscr{C}_K}= X^{N}.
 \eal
 \ee

\item The {\em cluster-input-equivalence} condition defined in \cite{Cluster_synchronization_Belykh_2008} holds. This implies that the following edge weight sums are equal:
 for any two nodes 
$X^i_{\mathscr{C}_r}, X^j_{\mathscr{C}_r}$, $(i,j)\in\mathscr{C}_r$, 
     \begin{align}\label{input_equi}
 \eta_{\mathscr{C}_r\mathscr{C}_s} \;:=\; 
 \sum_{k\in\mathscr{N}^{i}_{\mathscr{C}_s}} \gamma^{ik} = 
 \sum_{k\in\mathscr{N}^j_ {\mathscr{C}_s}}   \gamma^{jk}, 
    \end{align}
where $\mathscr{N}^i_{\mathscr{C}_s}$ denotes the indices of the neighbors of node $i$
 which are in cluster $\mathscr{C}_s$. 
\end{enumerate}
\eass

\blem Under Assumption \ref{Assumption}, the $K-$cluster synchronization manifold, defined in Definition \ref{Cluster_Synchronization}, is invariant. 
\elem
 
 \bp
This follows by the cluster-input-equivalence condition, Equation~(\ref{input_equi}).
\ep

Next we provide sufficient conditions to show that $\mathscr{S}_K$ is (globally) stable, 
i.e., any solution of Equation~(\ref{general_network}) converges to $\mathscr{S}_K$. 

Recall that the network graph is $\mathcal {G}= (\mathcal V, \mathcal E)$. Denote the subgraph 
for the nodes in $\mathscr{C}_r$ by 
$\mathcal G_{\mathscr{C}_r} = (\mathcal V_{\mathscr{C}_r}, \mathcal E_{\mathscr{C}_r})$. Then
 \[\mathcal G = \lt(\bigcup_{r=1}^K\; \mathcal G_{\mathscr{C}_r}\rt) \bigcup \bar{\mathcal G},\]
where $\bar{\mathcal G} = (\mathcal V, \mathcal E \setminus\cup_r\; \mathcal E_{\mathscr{C}_r})$ is 
the graph describing connections among the clusters ${\mathscr{C}_r}$. 

Let $\mathcal L_{\mathscr{C}_r}$ 
denote the Laplacian matrix of $\mathcal G_{\mathscr{C}_r}$ with eigenvalues 
$0= \lambda^{(1)}_{\mathscr{C}_r}\leq\lambda^{(2)}_{\mathscr{C}_r}\leq\ldots\leq\lambda^{(c_r)}_{\mathscr{C}_r}$ 
and $\bar{\mathcal L}$ 
denote the Laplacian matrix of $\bar{\mathcal G}$ with eigenvalues 
$0= \bar\lambda^{(1)}\leq\bar\lambda^{(2)}\leq\ldots\leq\bar\lambda^{(N)}$. 
In the special case of $K=1$, we set $\bar\lambda^{(2)}=0.$ 
Then $\mathcal{L}$, the Laplacian matrix of $\mathcal{G}$, can be written as follows:
\be{}
\mathcal{L} =\mathcal{L}_{\mathscr C} + \bar{\mathcal{L}}, 
\ee
where $\mathcal{L}_{\mathscr C}$ is a block diagonal matrix with the form:
\be{}
\mathcal{L}_{\mathscr C} = \left(\begin{array}{ccc}
\mathcal L_{\mathscr{C}_1} &  &  \\
 & \ddots &  \\ 
 &  & \mathcal L_{\mathscr{C}_K}\end{array}\right).
\ee
With these definitions, Equation~(\ref{general_network_closed_form1}) can be 
written as 
\begin{align}\label{general_network_closed_form2}
\dot{X} (t)&= \mathcal{F}(X(t), t) - (\mathcal L_{\mathscr C}\otimes \D) X(t) -  (\bar{\mathcal{L}}\otimes \D) X(t). 
\end{align}

\bthm\label{cluster_sychronization_main_result}
Consider Equation~(\ref{general_network}), or equivalently 
Equation~(\ref{general_network_closed_form2}), with Assumption \ref{Assumption}, and let 
\be{theorem_mu}
\mu := 
\max_{{r=1,\ldots,K}}\sup_{(x,t)\in V\times[0,\infty)}
 \M_{2,P} \lt[J_{F_{\mathscr{C}_r}}(x,t)-\lt(\lambda^{(2)}_{\mathscr{C}_r} + \bar\lambda^{(2)}\rt) \D\rt],
\ee
where $P\in\r^{n\times n}$ is a positive definite matrix chosen such that $P^2 \D+\D P^2$ is positive semidefinite. Then, for any solution $X$ of Equation~(\ref{general_network}) that remains in $V^N$, there exists $\bar X(t)$ such that 
\be{theorem_main_inequality}
\|X(t)- \bar X(t)\|_{2,\mathcal{P}} \leq e^{\mu t} \|X(0)- \bar X(0)\|_{2,\mathcal{P}},  
\ee
where $\mathcal{P} = I_N\otimes P^2$
and $\|\cdot\|_{2,\mathcal{P}}$ is a $\mathcal{P}$-weighted $L^2$ norm on $\r^{nN}$, defined by
\[
\|x\|_{2,\mathcal{P}}:=\left\|\left(\lt\|P^2x^1\rt\|_2, \ldots, \lt\|P^2x^N\rt\|_2\right)^T\right\|_2,
\]
for any $x=\lt({x^1}^T,\ldots, {x^N}^T\rt)^T\in\r^{nN}$. 
In particular, if $\mu<0$, then for any pair of nodes $i,j \in {\mathscr{C}_r}$, $X_{\mathscr{C}_r}^i$ and $X_{\mathscr{C}_r}^j$ satisfy
\[X_{\mathscr{C}_r}^i(t) - X_{\mathscr{C}_r}^j(t) \to 0 \quad \mbox{as $t\to\infty$}. \]
\ethm

\bremark
Theorem~\ref{cluster_sychronization_main_result} provides a sufficient  condition for cluster synchronization that depends on the dynamics of each isolated cluster $J_{F_{\mathscr{C}_r}}$, the diffusion matrix $D$, the structure  $\lambda^{(2)}_{\mathscr{C}_r}$ of each subgraph  $\mathcal{G}_{\mathscr{C}_r}$ describing connections among the nodes in cluster ${\mathscr{C}_r}$,
and the structure  $\bar\lambda^{(2)}$ of the subgraph  $\bar{\mathcal{G}}$ describing connections among the clusters. 
Proposition~\ref{synchronization_Arcak} is a special case of Theorem \ref{cluster_sychronization_main_result} when $K=1$ and $\bar\lambda^{(2)}=0$.
One can still apply Proposition~\ref{synchronization_Arcak} to $K>1$ clusters to show cluster synchronization. However, Theorem \ref{cluster_sychronization_main_result} provides a weaker sufficient condition for cluster synchronization. 
\eremark 

In the proof of Theorem \ref{cluster_sychronization_main_result}, we need the following key lemmas. We first state the Courant-Fischer minimax Theorem, from \cite{Horn_Johnson}. 

\begin{lemma}\label{poincare_discrete}
Let $L$ be a positive semidefinite matrix in $\mathbb{R}^{l\times l}$. Let $\lambda^{(1)}\leq \cdots\leq\lambda^{(l)}$ be $l$ eigenvalues with $e^1, \cdots, e^l$ corresponding normalized orthogonal eigenvectors. For any $v\in \mathbb{R}^l$, if $v^Te^j=0$ for $1\leq j\leq k-1$, then \[v^T Lv\geq \lambda^{(k)} v^Tv.\]
\end{lemma}

\begin{lemma}\label{Lyapanov-inequality}\cite[Lemma 3]{synchronization2014_journal}
Suppose that $P$ is a positive definite matrix and $A$ is an arbitrary matrix. If $\mu_{2,P}[A]= \mu$, then $P^2A+A^TP^2\leq 2\mu P^2$. 
\end{lemma}

\textbf{Proof of Theorem \ref{cluster_sychronization_main_result}} 

 Let $w:=X-\bar X$, where 
 \[
 X=\lt({X_{\mathscr{C}_1}^1}^T, \ldots, {X_{\mathscr{C}_1}^{c_1}}^T, 
 \ldots,
 {X_{\mathscr{C}_K}^1}^T, \ldots, {X_{\mathscr{C}_K}^{c_K}}^T\rt)^T,
 \] 
 is a solution of $(\ref{general_network})$ and 
  \[
  \bar X=\lt(\lt(\mathbf{1}_{c_1}\otimes x_1\rt)^T, \ldots, \lt(\mathbf{1}_{c_K}\otimes x_K\rt)^T \rt)^T,
  \] 
 with $x_r := \frac{1}{c_r} \sum_{i=1}^{c_r}X_{\mathscr{C}_r}^{i}$ and $\mathbf{1}_{c_r}\in\r^{c_r}$ is a vector of ones. Let $w=\lt(w_1^T, \ldots, w_K^T\rt)^T$, where $w_r := \lt((X_{\mathscr{C}_r}^1- x_r)^T, \ldots, (X_{\mathscr{C}_r}^{c_r}-  x_r)^T\rt)^T\in\r^{c_rn}$, and define 
\[
\Phi(w)\;:=\;  \displaystyle\frac{1}{2} w^T \mathcal P w
 = \displaystyle\frac{1}{2}\sum _{r=1}^K w_r^T \left(I_{c_r}\otimes P^2\right)w_r\;.
\]
 Since $\Phi(w)=\displaystyle\frac{1}{2}\|\mathcal Pw\|^2_2$, to prove (\ref{theorem_main_inequality}), it suffices to show that 
 \[\displaystyle\frac{d}{dt}\Phi(w)\leq 2\mu\Phi(w).\]  
Let
\[\mathcal{F}(X,t) = \lt(
F^T_{\mathscr{C}_1} (X^1_{\mathscr{C}_1}, t), \ldots,  F^T_{\mathscr{C}_1} (X^{c_1}_{\mathscr{C}_1}, t), 
\ldots, 
F^T_{\mathscr{C}_K} (X^1_{\mathscr{C}_K}, t), \ldots,  F^T_{\mathscr{C}_K} (X^{c_K}_{\mathscr{C}_K}, t)
\rt)^T,\]
and
\[\bar {\mathcal{F}}(X,t) = \lt(\lt(\mathbf{1}_{c_1}\otimes y_1\rt)^T, \ldots,\lt(\mathbf{1}_{c_K}\otimes y_K\rt)^T \rt)^T\quad
\mbox{where $\;y_r = \dis\frac{1}{c_r}\sum_{i=1}^{c_r}F_{\mathscr{C}_r} (X^i_{\mathscr{C}_r}, t)$}.
\] 
Standard calculations show that the derivative of $\Phi$ is as follows:
\be{Phi-dot_discrete}
\bal
\displaystyle\frac{d\Phi}{dt}(w)
&=  w^T\mathcal P \lt(\mathcal{F}(X,t) - \bar {\mathcal{F}}(X,t)\rt)
-w^T\mathcal P(\mathcal{L}_{\mathscr{C}} \otimes D)w
- w^T \mathcal P(\bar{\mathcal L}\otimes D)\\
&=  w^T\mathcal P \lt(\mathcal{F}(X,t) -  \mathcal{F}(\bar X,t)\rt) 
+w^T\mathcal P \lt(\mathcal{F}(\bar X,t) - \bar {\mathcal{F}}(X,t)\rt)
-w^T\mathcal P(\mathcal{L}_{\mathscr{C}} \otimes D)w
- w^T \mathcal P(\bar{\mathcal L}\otimes D) w\\
&=  w^T\mathcal P \lt(\mathcal{F}(X,t) - \mathcal{F}(\bar X,t)\rt)
-w^T\mathcal P(\mathcal{L}_{\mathscr{C}} \otimes D)w
- w^T \mathcal P(\bar{\mathcal L}\otimes D) w\;.
\eal
\ee
In the second equation, we added and subtracted $w^T \mathcal P\mathcal{F}(\bar X,t)$, where $\mathcal{F}(\bar X,t)$ is written as 
\[\mathcal{F}(\bar X,t) =\lt(
\lt(\mathbf{1}_{c_1}\otimes F_{\mathscr{C}_1} (x_1, t)\rt)^T,\ldots, 
\lt(\mathbf{1}_{c_K}\otimes F_{\mathscr{C}_K} (x_K, t)\rt)^T
\rt)^T.\]
The last equality holds because $w_r^T (\mathbf{1}_{c_r}\otimes I_n) =0$ implies that 
\beqn
w^T\mathcal P \lt(\mathcal{F}(\bar X,t) -  \bar {\mathcal{F}}(X,t)\rt) 
&=& \sum_{r=1}^K w_r^T \lt(I_{c_r}\otimes P^2\rt)  \lt( \mathbf{1}_{c_r}\otimes \lt(F^T_{\mathscr{C}_r} (x_r, t)-y_r^T\rt)\rt)\\
&=& \sum_{r=1}^K w_r^T \lt(\mathbf{1}_{c_r} \otimes P^2 \lt( F^T_{\mathscr{C}_r} (x_r, t)-y_r^T\rt)\rt)\\
 &=& \sum_{r=1}^K w_r^T  \lt(\mathbf{1}_{c_r}\otimes I_n\rt)P^2 \lt(F^T_{\mathscr{C}_r} (x_r, t)- y_r^T \rt)\\
&=&0.
\eeqn 

 \textbf{Step 1.} We show that  
 \be{step1}
 -w^T\mathcal P(\mathcal{L}_{\mathscr{C}} \otimes D)w 
 \leq
 -\sum_{r=1}^K\lambda_{\mathscr{C}_r}^{(2)} w_r^T\lt(I_{c_r}\otimes P^2D\rt)w_r\;. 
 \ee
 Since $P^2D+DP^2$ is positive semidefinite, Cholesky decomposition yields an upper triangular matrix $M$ such that $P^2D+DP^2=2M^TM$. For any $r=1, \ldots, K$, 
\begin{equation*}
\begin{aligned}
 -w_r^T\left(I_{c_r}\otimes P^2\right)(\mathcal{L}_{\mathscr{C}_r} \otimes D)w_r
 & =-w_r^T\left(\mathcal{L}_{\mathscr{C}_r}\otimes P^2 D\right)w_r\\
 &= -\frac{1}{2} w_r^T\left(\mathcal{L}_{\mathscr{C}_r}\otimes \lt(P^2 D+DP^2\rt)\right)w_r\\
 &= -w_r^T\left(\mathcal{L}_{\mathscr{C}_r}\otimes \lt(M^TM\rt)\right)w_r\\
 &= -w_r^T\lt(I_{c_r}\otimes M^T\rt)\left(\mathcal{L}_{\mathscr{C}_r}\otimes I_n\right)\left(I_{c_r}\otimes M\right) w_r\\
 &\leq-  \lambda_{\mathscr{C}_r}^{(2)} \left((I_{c_r}\otimes M)w_r\right)^T(I_{c_r}\otimes M)w_r\\
 &=-\lambda_{\mathscr{C}_r}^{(2)} w_r^T\lt(I_{c_r}\otimes M^T M\rt)w_r\\
 &=-\lambda_{\mathscr{C}_r}^{(2)} w_r^T\lt(I_{c_r}\otimes P^2D\rt)w_r\;.
 \end{aligned}
\end{equation*}
Note that the inequality holds by Lemma \ref{poincare_discrete}. To apply Lemma \ref{poincare_discrete}, we need to show that
 \[\lt(\left(I_{c_r}\otimes M\right) w_r\rt)^T(\mathbf{1}_{c_r}\otimes I_n) = 0.\]  
By definition of $w_r$, $w_r^T\mathbf{1}_{nc_r}=0$ and hence 
\beqn
&&\lt(\left(I_{c_r}\otimes M\right) w_r\rt)^T(\mathbf{1}_{c_r}\otimes I_n)
= w_r^T \left(I_{c_r}\otimes M^T\right) (\mathbf{1}_{c_r}\otimes I_n)
= w_r^T\lt(\mathbf{1}_{c_r}\otimes M^T\rt)\\
&&\qquad\qquad\qquad = \sum_{i=1}^{c_r} (X_{\mathscr{C}_r}^i- x_r)^TM^T
=\lt(\sum_{i=1}^{c_r} (X_{\mathscr{C}_r}^i- x_r)^T\rt)M^T
= 0.
\eeqn 
Both $\mathcal P$ and $\mathcal{L}_{\mathscr{C}}$ are block diagonal with blocks of same sizes, 
$c_1,\ldots, c_K$, so we have:
\beqn 
-w^T\mathcal P(\mathcal{L}_{\mathscr{C}} \otimes D)w 
= -\sum_{r=1}^Kw_r^T\left(I_{c_r}\otimes P^2\right)(\mathcal{L}_{\mathscr{C}_r} \otimes D)w_r
\leq -\sum_{r=1}^K\lambda_{\mathscr{C}_r}^{(2)} w_r^T\lt(I_{c_r}\otimes P^2D\rt)w_r\;.
\eeqn 

 \textbf{Step 2.} We show that
  \be{step2}
  -w^T \mathcal P(\bar{\mathcal L}\otimes D) w
  \leq
- \sum_{r=1}^K \bar \lambda^{(2)}w_r^T\lt(I_{c_r}\otimes P^2D\rt)w_r\;.
\ee 
The proof is analogous to the previous step. 
\beqn
-w^T \mathcal P(\bar{\mathcal L}\otimes D) w 
&=&-w^T \lt(I_N\otimes P^2\rt)(\bar{\mathcal L}\otimes D) w\\
&=&-w^T \lt(\bar{\mathcal L}\otimes P^2D\rt) w\\
&=& -\frac{1}{2} w^T \lt(\bar{\mathcal L}\otimes \lt(P^2D+DP^2\rt)\rt) w\\
&=& - w^T \lt(\bar{\mathcal L}\otimes M^TM\rt) w\\
&=& - w^T\lt(I_N\otimes M^T\rt) \lt(\bar{\mathcal L}\otimes I_n \rt) \lt(I_N\otimes M \rt)w\\
&\leq& -  \bar\lambda^{(2)} w^T\lt(I_N\otimes M^T\rt) \lt(I_N\otimes M \rt)w\\
&=& -  \bar\lambda^{(2)} w^T\lt(I_N\otimes M^TM\rt)w\\
&=& -  \bar\lambda^{(2)} w^T\lt(I_N\otimes P^2D\rt)w\\
&=&  -\sum_{r=1}^K\bar \lambda^{(2)} w_r^T\lt(I_{c_r}\otimes P^2D\rt)w_r\;.
  \eeqn 

\textbf{Step 3.} We show that 
\be{step3}
w^T\mathcal{P}(\mathcal{F}(X, t)-\mathcal{F}(\bar X, t))= \sum_{r=1}^K \displaystyle\sum_{i=1}^{c_r}\displaystyle\int_0^1 (X_{\mathscr{C}_r}^i- x_r)^T P^2 
J_{F_{\mathscr{C}_r}} \lt(x_r +\tau(X_{\mathscr{C}_r}^i- x_r)\rt)(X_{\mathscr{C}_r}^i- x_r)\;d\tau.
\ee 
Note that 
$
 w^T\mathcal{P}(\mathcal{F}(X, t)-\mathcal{F}(\bar X, t))
 =
 \sum_{r=1}^K w_r^T\lt(I_{c_r}\otimes P^2\rt) \tilde{\mathcal{F}}_r (X_{\mathscr{C}_r})\;, 
$
where 
\[
\tilde{\mathcal{F}}_r (X_{\mathscr{C}_r}) 
=\lt( F^T_{\mathscr{C}_r} (X^1_{\mathscr{C}_r}, t)- F^T_{\mathscr{C}_r} (x_r, t),
 \ldots, 
 F^T_{\mathscr{C}_r} (X^{c_r}_{\mathscr{C}_r}, t)- F^T_{\mathscr{C}_r} (x_r, t)
 \rt)^T.
\]
 By the Mean Value Theorem for integrals, for any $r=1,\ldots, K$, 
 \begin{equation*}
\begin{aligned}
w_r^T\lt(I_{c_r}\otimes P^2\rt) \tilde{\mathcal{F}}_r (X_{\mathscr{C}_r})
&=\displaystyle\sum_{i=1}^{c_r} 
(X_{\mathscr{C}_r}^i- x_r)^T P^2 \lt(F_{\mathscr{C}_r} (X^i_{\mathscr{C}_r}, t)- F_{\mathscr{C}_r} (x_r, t)\rt)\\
&=\displaystyle\sum_{i=1}^{c_r}\displaystyle\int_0^1 (X_{\mathscr{C}_r}^i- x_r)^T P^2 
J_{F_{\mathscr{C}_r}} \lt(x_r +\tau(X_{\mathscr{C}_r}^i- x_r)\rt)(X_{\mathscr{C}_r}^i- x_r)\;d\tau.
\end{aligned}
\end{equation*}
Adding over $r$, $r=1,\ldots, K$, we obtain Equation~(\ref{step3}). 

Note that the sum of the left hand side of Equations~(\ref{step1})-(\ref{step3}), is equal to $\frac{d\Phi}{dt}$. 
Combining Steps 1-3, we have shown that 
\[\frac{d\Phi}{dt} \leq \sum_{r=1}^K \phi_r,\]
 where for any $r=1,\ldots, K$, 
\be{phi_r}
\bal
\phi_r&:= w_r^T\lt(I_{c_r}\otimes P^2\rt) \tilde{\mathcal{F}}_r (X_{\mathscr{C}_r}) 
- w_r^T\lt(I_{c_r}\otimes P^2\rt) \lt(I_{c_r}\otimes \lambda_{\mathscr{C}_r}^{(2)}D\rt) w_r
- w_r^T\lt(I_{c_r}\otimes P^2\rt) \lt(I_{c_r}\otimes \bar\lambda^{(2)}D\rt) w_r\\
&=\displaystyle\sum_{i=1}^{c_r}\displaystyle\int_0^1 (X_{\mathscr{C}_r}^i- x_r)^T P^2 
\lt[J_{F_{\mathscr{C}_r}} \lt(x_r+\tau(X_{\mathscr{C}_r}^i- x_r)\rt) -  \lambda_{\mathscr{C}_r}^{(2)}D -  \bar\lambda^{(2)}D\rt]
(X_{\mathscr{C}_r}^i- x_r)\;d\tau\\
 &\leq \displaystyle\sum_{r=1}^{c_r}\displaystyle\frac{2\mu}{2}\displaystyle\int_0^1d\tau \;(X_{\mathscr{C}_r}^i- x_r)^T P^2(X_{\mathscr{C}_r}^i- x_r) \\
  &=\displaystyle\frac{2\mu}{2} \;w_r^T\lt(I_{c_r}\otimes P^2\rt)w_r\;.
\eal
\ee
The inequality holds by applying Lemma \ref{Lyapanov-inequality} to Equation (\ref{theorem_mu}): we obtain, for any $r=1,\ldots, K$, and any $(x,t)\in V\times[0,\infty)$, 
\beqn
P^2\lt[J_{F_{\mathscr{C}_r}}(x,t)-\lt(\lambda^{(2)}_{\mathscr{C}_r} + \bar\lambda^{(2)}\rt) \D\rt]+
\lt[J_{F_{\mathscr{C}_r}}^T(x,t)-\lt(\lambda^{(2)}_{\mathscr{C}_r} + \bar\lambda^{(2)}\rt) \D\rt] P^2
\leq2\mu P^2. 
\eeqn
Summing both sides of Equation~(\ref{phi_r}) over $r$, for $r=1,\ldots, K$, we obtain the desired result, $\frac{d\Phi}{dt}(w)\;\leq\; 2\mu\Phi(w).$
\qed

\section{Network Reduction}
\label{Network_Reduction}

We now outline a procedure for reducing a cluster synchronized network to a smaller network of synchronized states, commonly referred to as the {\em quotient network}. Quotient networks have been applied to find reductions of cluster synchronized networks with slight parameter mismatches in the $F^i$ \cite{sorrentino2016approximate}. In a cluster synchronized network, oscillators in the same cluster will have homogeneous dynamics after the initial transients. The longer-term dynamics of the network can thus be reduced to those of a network where each node corresponds to a cluster. This reduction loses no information about the long-term dynamics and can be implemented to simplify analysis.
   
 \blem Under Assumption \ref{Assumption}, the dynamics of Equation~(\ref{general_network}) on $\mathscr{S}_K$ can be described as follows. 
\begin{align}\label{synchronized_state}
\dot{{X}}_{\mathscr{C}_r} (t)&= F_{ \mathscr{C}_r}\lt( {X}_{\mathscr{C}_r}(t),t\rt) 
+ \sum_{\mathscr{C}_s\in\mathscr{N}_{\mathscr{C}_r}}  \eta_{\mathscr{C}_r\mathscr{C}_s} \D\left( {X}_{\mathscr{C}_s}(t)-  {X}_{\mathscr{C}_r}(t)\right) \qquad r=1,\ldots, K, 
    \end{align}
where $\mathscr{N}_{\mathscr{C}_r}$ denotes the set of all the clusters that are connected to ${\mathscr{C}_r}$, as in Assumption \ref{Assumption}. 
\elem
 
 \bp
This follows by the cluster-input-equivalence condition, Equation \eqref{input_equi}.
\ep

The simplified dynamics in Equation~(\ref{synchronized_state}) represent a powerful tool for facilitating analysis of the dynamics of cluster synchronized systems.

\section{Application to networks of heterogeneous FitzHugh-Nagumo neuronal oscillators}
\label{Application}

In this section, we apply Theorem \ref{cluster_sychronization_main_result} to a network of $N$ FitzHugh-Nagumo (FN) neuronal oscillators with graph $\mathcal{G}$.

Let $(y^i,z^i)^T\in\r^2$ be the state of oscillator $i$ and $I^i$ be the external input to oscillator $i$, for $i=1,\ldots, N$.  $y^i$ and $z^i$ represent the membrane potential and the recovery variable, repectively, and $I^i$ the input current.  
The FN dynamics are
\be{FN_network}
\bal
\dot{y}^i &= f^i(y^i) -z^i + I^i  + \gamma\sum_{j\in \mathcal{N}^i} \gamma^{ij}(y^j-y^i),\\
\dot{z}^i &= \epsilon^i (y^i-b^i z^i), 
\eal
\ee
where $f^i$ is a cubic function, $f^i(y) = y-\frac{y^3}{3}-a^i$, $\gamma, a^i, b^i>0$, $0< \epsilon^i\ll1$, and
$\mathcal{N}^i$ denotes the set of all the neighbors of node $i$ in the network. 
Using the notation of Theorem \ref{cluster_sychronization_main_result}, 
$n=2$, $X^i= (y^i,z^i)^T$, 
$F^i(X^i,t) = \left(f^i(y^i) -z^i + I^i, \epsilon^i (y^i-b^i z^i)\right)^T$,  
$\D= \diag(\gamma, 0)$ is the diffusion matrix, 
and the $\gamma^{ij}$ are the edge weights on the graph $\mathcal{G}$. 

Assume that there exist $K\geq1$ clusters $\mathscr{C}_1, \ldots, \mathscr{C}_K$ of FN oscillators such that $a^i = a_{\mathscr{C}_r}$, $b^i= b_{\mathscr{C}_r}$, $\e^i=\e_{ \mathscr{C}_r}$, and $I^i = I_{\mathscr{C}_r}$ for all FN oscillators $i\in\mathscr{C}_r$ and all clusters $r= 1, \ldots, K$. 

 In what follows we show that, for $K=1$ cluster, if $\gamma\lambda^{(2)}>1$, then Equation~(\ref{FN_network}) synchronizes. More generally, if $K>1$, and for all $r=1,\ldots, K$, $\e_{ \mathscr{C}_r}=\e$, and $\gamma\lambda^{(2)}_{\mathscr{C}_r} +\gamma \bar\lambda^{(2)}>1$, then Equation~(\ref{FN_network}) converges to its $K-$cluster synchronization manifold. 

\begin{proposition} \label{application_FN_proposition}
Consider Equation~(\ref{FN_network}), with Assumption~\ref{Assumption}. For all $r=1,\ldots, K$, let 
 \[\gamma >\dis\frac{1+\a_r}{\lambda^{(2)}_{\mathscr{C}_r}+ \bar\lambda^{(2)}},\]    
 where $\a_r=\frac{\lt(\e_{\mathscr{C}_r} p - 1/p\rt)^2}{4b_{\mathscr{C}_r}\e_{\mathscr{C}_r}}$ and $p=\max_{r}\frac{1}{\sqrt{\e_{\mathscr{C}_r}}}$. 
 Then for any pair of FN oscillators $\{(y^i,z^i)^T,(y^j,z^j)^T\}$
such that $(i,j)\in\mathscr{C}_r$, 
\[y^i(t) - y^j(t) \to 0, \quad z^i(t) - z^j(t) \to 0, \quad \mbox{as $t\to\infty$}. \]
\end{proposition}

\begin{proof}
To apply Theorem \ref{cluster_sychronization_main_result}, we find a positive definite matrix $P$ such that $P^2D+DP^2$ is positive semidefinite and 
\[\mu:= \max_r\sup_{(y,z)^T\in\r^2}\M_{2,P}\lt[ J_{F_{\mathscr{C}_r}}(y,z) - \lt(\lambda^{(2)}_{\mathscr{C}_r} + \bar\lambda^{(2)}\rt)D\rt] <0.\]
Let $P=\diag(1, p)$ so that $P^2 D+DP^2 = \diag(2\gamma, 0)$, which is positive semidefinite. Then 
\be{}
\bal
  \M_{2,P}\lt[ J_{F_{\mathscr{C}_r}}(y,z)- \lt(\lambda^{(2)}_{\mathscr{C}_r} + \bar\lambda^{(2)}\rt)D\rt]
 &=\M_{2}\lt[P\lt( J_{F_{\mathscr{C}_r}}(y,z)- \lt(\lambda^{(2)}_{\mathscr{C}_r} + \bar\lambda^{(2)}\rt) D\rt)P^{-1} \rt]\\
&=\lambda_{\max}\lt[\left(\begin{array}{cc} 
1-y^2-\gamma\lambda^{(2)}_{\mathscr{C}_r}-\gamma \bar\lambda^{(2)}  & \frac{\e_{\mathscr{C}_r} p}{2} - \frac{1}{2p}
\\  \frac{\e_{\mathscr{C}_r} p}{2} - \frac{1}{2p} & -b_{\mathscr{C}_r}\epsilon_{\mathscr{C}_r}
\end{array} \right) \rt].  
\eal
\ee
To see this recall that $\M_{2,P}[A] = \M_2[PAP^{-1}]$, and, by Remark \ref{explicit_value_LN}, $\M_2[A] = \lambda_{\max} \lt[\frac{A+A^T}{2}\rt]$, where $\lambda_{\max}[B]$ denotes the largest eigenvalue of $B$. Note that the matrix shown in the second line, call it $\mathcal B$, is the symmetric part of $P\lt( J_{F_{\mathscr{C}_r}}(y,z)- \lt(\lambda^{(2)}_{\mathscr{C}_r} + \bar\lambda^{(2)}\rt) D\rt)P^{-1}$. 
Standard calculations show that if $\gamma >\frac{1+\a_r}{\lambda^{(2)}_{\mathscr{C}_r}+ \bar\lambda^{(2)}}\geq \frac{1}{\lambda^{(2)}_{\mathscr{C}_r}+ \bar\lambda^{(2)}}$
then the trace and the determinant of $\mathcal B$ satisfy
\beqn 
&\mbox{Tr} = 1-y^2- \gamma\lambda^{(2)}_{\mathscr{C}_r}-\gamma \bar\lambda^{(2)} - b_{\mathscr{C}_r}\epsilon_{\mathscr{C}_r}<0, 
\quad\mbox{Det} = -b_{\mathscr{C}_r}\epsilon_{\mathscr{C}_r} \lt( 1-y^2- \gamma\lambda^{(2)}_{\mathscr{C}_r}-\gamma \bar\lambda^{(2)} +\a_r \rt)>0.  
\eeqn 
Therefore, $\lambda_{\max}[\mathcal B]<0$ and Theorem \ref{cluster_sychronization_main_result} yields the desired result. 

\end{proof}
 
\bremark
In Proposition \ref{application_FN_proposition}:
\begin{enumerate}
\item If we assume that, for all $r=1,\ldots, K$, $\e_{\mathscr{C}_r}= \e$, then
$\a_r=0$ and we obtain a smaller lower bound for $\gamma$, namely
\[
\gamma > \dis\frac{1} {\lambda^{(2)}_{\mathscr{C}_i}+ \bar\lambda^{(2)}}.
\]
\item Non-diagonal $P$ does not give a smaller lower bound for $\gamma$.
\item Theorem~\ref{cluster_sychronization_main_result} can be used to derive an analogous condition for a network of FN oscillators with time varying parameters.
\end{enumerate}
\eremark

\bremark In the previous work \cite{davison_sync_2016}, we showed that for $K=1$, if $\gamma \geq\frac{1+ \epsilon+{\beta^2}/{3}}{\lambda^{(2)}}$, where $\beta$ is the ultimate bound for the $y$ variable, then Equation~(\ref{FN_network}) synchronizes. By Proposition \ref{application_FN_proposition} we have found a smaller lower bound for $\gamma$, $\gamma >\frac{1}{\lambda^{(2)}}$, that guarantees synchronization. 
\eremark

\subsection {Examples }

\begin{example}
\em{
In this example, we consider a network of 17 FN oscillators (shown in the left panel of Figure~\ref{CIE_Obeyance_Illustr}), wherein each oscillator has the dynamics associated with one of three different clusters: (i) $\mathscr{C}_1$ is a cluster of six oscillators (represented by orange circles) with $a=0.1$, $b = 0.1$, $\epsilon=0.08$ and $I=0.9$; (ii) $\mathscr{C}_2$ is a cluster of seven oscillators (represented by green squares) with $a=0.5$, $b = 0.7$, $\epsilon=0.08$ and $I=3.0$; and (iii) $\mathscr{C}_3$ is a cluster of four oscillators (represented by blue triangles) with $a=0.9$, $b = 0.3$, $\epsilon=0.08$ and $I=0.1$. For this network, the second smallest eigenvalue of the Laplacian of each corresponding subgraph can be computed as $\lambda^{(2)}_{\mathscr{C}_1} = 1.83$, $\lambda^{(2)}_{\mathscr{C}_2} = \lambda^{(2)}_{\mathscr{C}_3} = 2.00$ and $\bar\lambda^{(2)} = 13.10$. Then from Proposition~\ref{application_FN_proposition}, we can conclude that the clusters will synchronize whenever $\gamma > 0.067$, since the cluster-input-equivalence condition \eqref{input_equi} holds true. As shown in Figure~\ref{CIE_Obeyance_Illustr}, the network indeed displays fast convergence to cluster synchronization with $\gamma = 0.120$.
\begin{figure}[ht!]
	\centering
	\includegraphics[width =0.90\textwidth]{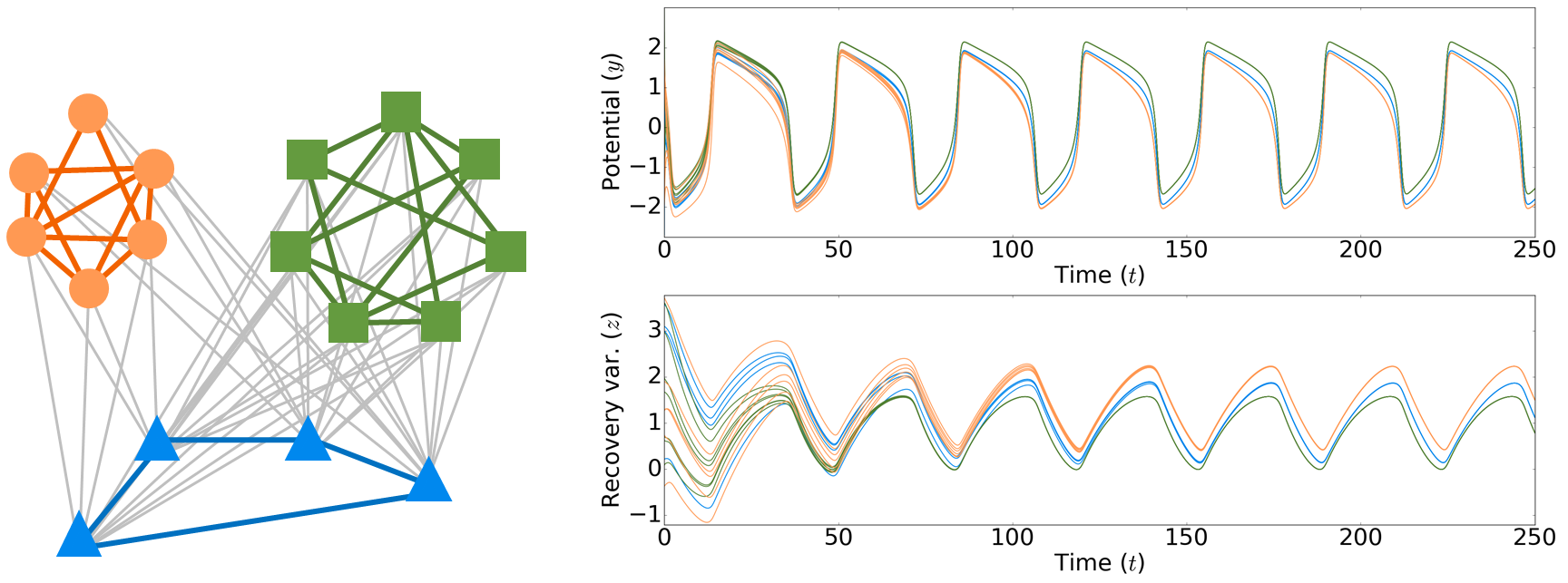}
	\caption{Cluster synchronization in a network of 17 heterogeneous FN oscillators.}
	\label{CIE_Obeyance_Illustr}
\end{figure}

However, when we introduce heterogeneity within the third cluster $\mathscr{C}_3$, e.g., by changing parameter values to $a=0.8$, $b = 0.9$ and the external input to $I=0.7$ for two of the four oscillators (these are now represented by magenta pentagons in Figure~\ref{Single_Agent_Cluster_Sync}), the blue cluster breaks into two clusters, each with two oscillators. As a result, the orange cluster no longer satisfies the cluster-input-equivalence condition \eqref{input_equi} unless it too breaks into two clusters of three oscillators each (shown in light and dark orange in Figure~\ref{Single_Agent_Cluster_Sync}).  By Proposition~\ref{application_FN_proposition} the condition for cluster synchronization is again $\gamma > 0.067$; however, now there are five clusters as illustrated in Figure~\ref{Single_Agent_Cluster_Sync} for $\gamma = 0.120$. 
\begin{figure}[h!]
	\centering
	\includegraphics[width =0.90\textwidth]{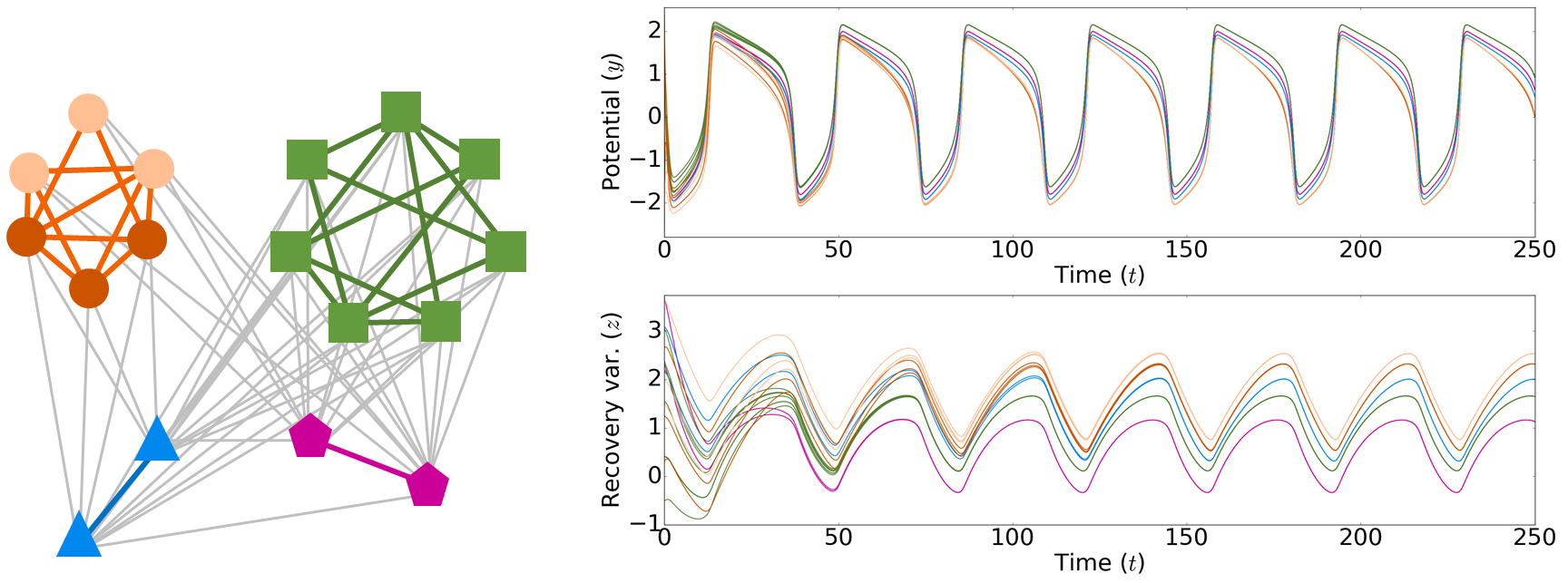}
	\caption{Emergence of new clusters, for a total of five, in the same network (as in Figure~\ref{CIE_Obeyance_Illustr}) of 17 FN oscillators as a result of modifying the dynamics of the two oscillators shown in magenta.}
	\label{Single_Agent_Cluster_Sync}
\end{figure}
}
\end{example}

\begin{example}
\em{
To illustrate the critical role of the cluster-input-equivalence condition \eqref{input_equi} in cluster synchronization, we consider a slightly perturbed version of the network shown in Figure~\ref{CIE_Obeyance_Illustr} by removing some connections between clusters. Removal of connections between the clusters leads to a lower connectivity of the subgraph $\bar{\mathcal{G}}$; for the network connections in Figure~\ref{CIE_Violation_Illustr}, $\bar\lambda^{(2)} = 6.81$. Although everything else remain same as the original network considered in the previous example, this perturbation leads to a violation of the cluster-input-equivalence condition. As a result, the network fails to achieve cluster synchronization even when $\gamma = 0.120 > 0.116$ satisfies the sufficient condition (Fig~\ref{CIE_Violation_Illustr}).
\begin{figure}[th!]
	\centering
	\includegraphics[width =0.90\textwidth]{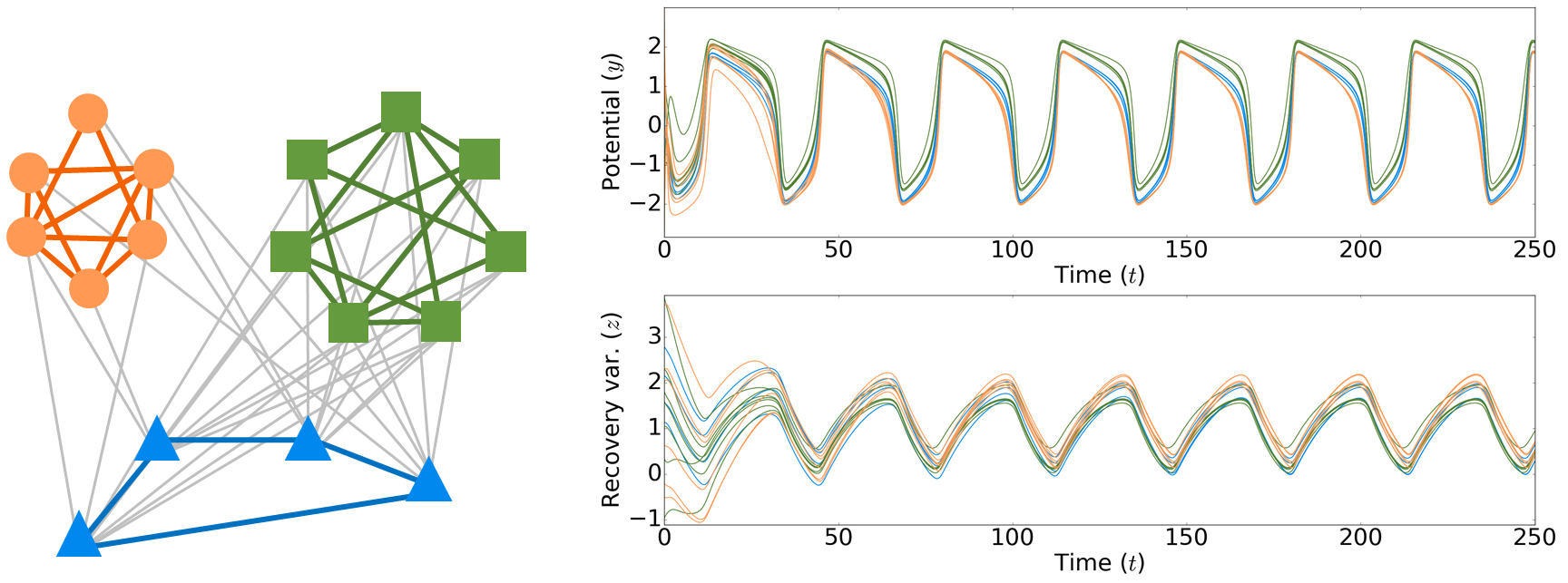}
	\caption{Collapse of cluster synchronization in a network of 17 heterogeneous FN oscillators.  $\gamma = 0.120$ as in Figures~\ref{CIE_Obeyance_Illustr} and \ref{Single_Agent_Cluster_Sync}, but the cluster-input-equivalence condition is no longer satisfied.}
	\label{CIE_Violation_Illustr}
\end{figure}
}
\end{example}

\begin{example}
\em{
In this example we consider a large network with 200 FN oscillators (refer to the left panel of Figure~\ref{Cluster_Sync_with_IntraCluster} for a representation of this network) obtained through interconnection of a complete graph (of size 100) with a star graph (of size 100). By connecting each node in the complete graph with a unique node in the star graph and edge of weight $0.25$, we ensure that the cluster-input-equivalence condition holds true. The FN oscillators ($\mathscr{C}_1$) in the complete graph (represented by magenta squares) have $a=0.9$, $b = 0.5$, and they are excited with an external current $I=2.0$. On the other hand, the FN oscillators ($\mathscr{C}_2$) in the star graph (represented by green triangles) have $a=0.7$, $b = 0.8$, and they are excited with an external current $I=0.3$. Also, we let $\epsilon = 0.08$ for each of these 200 oscillators. For this network $\lambda^{(2)}_{\mathscr{C}_1} =100$, $\lambda^{(2)}_{\mathscr{C}_2} = 0.04$ and $\bar\lambda^{(2)} = 0$. By choosing a diffusion constant $\gamma = 0.02$ such that $\gamma > 1/\big(\lambda^{(2)}_{\mathscr{C}_1} + \bar\lambda^{(2)} \big)$ but $\gamma < 1/\big(\lambda^{(2)}_{\mathscr{C}_2} + \bar\lambda^{(2)} \big)$ we do not obey the sufficient condition. However, numerical simulation (Figure~\ref{Cluster_Sync_with_IntraCluster}) shows that the magenta cluster ($\mathscr{C}_1$) synchronizes nevertheless as suggested by the fact that $\gamma > 1/\big(\lambda^{(2)}_{\mathscr{C}_1} + \bar\lambda^{(2)} \big)$ is satisfied. 
\begin{figure}[h]
	\centering
	\includegraphics[width =0.9\textwidth]{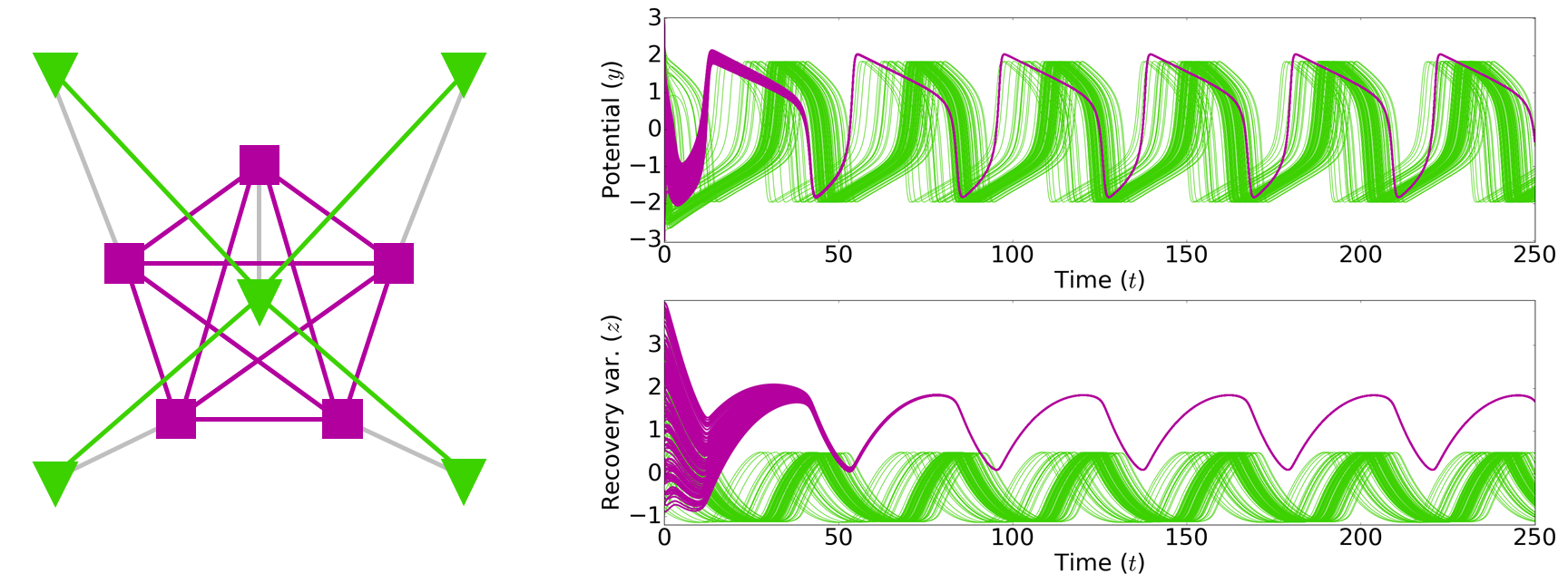}
	\caption{Synchronization of only one of two clusters in a large network of heterogeneous FN oscillators when the coupling strength takes an intermediate value.  There are 100 oscillators in one cluster connected through a star graph (green) and 100 oscillators in a second cluster connected through a complete graph (magenta). The network on the left illustrates the connections between clusters (in gray) in the case of 5 oscillators in each cluster.}
	\label{Cluster_Sync_with_IntraCluster}
\end{figure}
}
\end{example}

\section{Conclusion}

In this paper, we consider the patterns of synchronization that emerge in networks where individual nodes may have different intrinsic nonlinear dynamics. By adopting an approach based on contraction theory \cite{Aminzare_thesis}, our work proposes a sufficient condition for cluster synchronization, and provides its characterization in terms of the within-cluster network structure and the across-cluster network structure. We also explore a necessary condition for cluster synchronization, namely the cluster-input-equivalence condition, and demonstrate that its violation can lead to collapse of cluster synchronization (Figure~\ref{CIE_Violation_Illustr}). Another key contribution of this work is an improvement on previous sufficient conditions for cluster synchronization \cite{davison_sync_2016} in networks of oscillators with heterogeneous intrinsic dynamics. We also noticed through numerical simulation that heterogeneity within a particular cluster can cause desynchronization in another cluster (Figure~\ref{Single_Agent_Cluster_Sync}). Building upon this observation, our future work will attempt to develop a framework for designing time-varying inputs that will lead to fission and subsequent fusion of clusters.

\section*{Acknowledgments}
%
This work was jointly supported by the National Science Foundation under NSF-CRCNS grant DMS-1430077 and the Office of Naval Research under ONR grant N00014-14-1-0635. This material is also based upon work supported by the National Science Foundation Graduate Research Fellowship under grant DGE-1656466. Any opinion, findings, and conclusions or recommendations expressed in this material are those of the authors and do not necessarily reflect the views of the National Science Foundation.

%
\footnotesize

\end{document}